\documentclass[12pt]{article} 

\usepackage[margin=1in]{geometry}
\usepackage[colorlinks,citecolor=blue,urlcolor=blue]{hyperref}

\usepackage{xfrac}
\usepackage{natbib}
\usepackage{dsfont}
\usepackage{amsmath}
\usepackage{booktabs}
\usepackage{mathrsfs}
\usepackage{amsfonts}
\usepackage{amssymb}
\usepackage{wrapfig}
\usepackage{mathtools}
\usepackage{graphicx}
\usepackage{amsthm}
\usepackage{fancyhdr}
\usepackage{enumitem}
\usepackage{algorithmic}
\usepackage{algorithm}
\usepackage{wrapfig}
\usepackage{booktabs}
\usepackage{multirow}
\usepackage{caption}
\usepackage{amsthm}
\usepackage{tikz}
\usepackage[toc,page]{appendix}
\usepackage{subcaption}
\usepackage{xspace}
\usepackage{setspace}
\usepackage{soul}

\newcommand{\ddr}{\mathrm{d}}
\newcommand{\e}{\mathrm{e}}

\newcommand{\R}{\mathds{R}}
\newcommand{\E}{\mathds{E}}

\newcommand{\simiid}{\stackrel{\mathrm{iid}}{\sim}}
\newcommand{\simind}{\stackrel{\mathrm{ind}}{\sim}}

\DeclareMathOperator{\Tr}{Tr}

\newcommand{\Xv}{\mathbf{X}}
\newcommand{\bv}{\mathbf{b}}
\newcommand{\Bm}{\mathbf{B}}
\newcommand{\Sm}{\mathbf{S}}
\newcommand{\xv}{\mathbf{x}}
\newcommand{\yv}{\mathbf{y}}

\newcommand{\tv}{\mathbf{t}}
\newcommand{\mv}{\mathbf{m}}
\newcommand{\thetav}{\boldsymbol{\theta}}
\newcommand{\lambdav}{\boldsymbol{\lambda}}
\newcommand{\muv}{\boldsymbol{\mu}}
\newcommand{\Sigmam}{\boldsymbol{\Sigma}}
\newcommand{\piv}{\boldsymbol{\pi}}
\newcommand{\Cm}{\mathbf{C}}

\newcommand{\Xm}{\mathbf{X}}
\newcommand{\diag}{\text{diag}}

\usepackage{mdframed}

\newtheorem{lemma}{Lemma}
\newtheorem{prop}{Proposition}

\newtheorem{theo}{Theorem}
\theoremstyle{definition}

\AtEndEnvironment{exmp}{\null\hfill\qedsymbol}

\allowdisplaybreaks

\usepackage{setspace}

\usepackage{authblk}

\title{\textbf{Dirichlet process mixtures \\ under affine transformations of the data}}

\author[1]{Julyan Arbel} 
\affil[1]{\small Universit\'e Grenoble Alpes, Inria, CNRS, Grenoble INP, LJK, 38000 Grenoble, France\\ \texttt{julyan.arbel@inria.fr}}
\author[2]{Riccardo Corradin}
\affil[2]{Department of Economics, Management and Statistics, University of Milano Bicocca, Italy\\ \texttt{riccardo.corradin@unimib.it} $\quad$ \texttt{bernardo.nipoti@unimib.it}}
\author[2]{Bernardo Nipoti}

\date{}

\begin{document}
\doublespacing
\maketitle
%
%
%
%
%
%


\begin{abstract}
Location-scale Dirichlet process mixtures of Gaussians (DPM-G) have proved extremely useful in dealing with density estimation and clustering problems in a wide range of domains. Motivated by an  astronomical application, in this work we address the robustness of DPM-G models to affine transformations of the data, a natural requirement for any sensible statistical method for density estimation and clustering. First, we devise a coherent prior specification of the model which makes posterior inference invariant with respect to affine transformations of the data. Second, we formalise the notion of asymptotic robustness under data transformation and show that mild assumptions on the true data generating process are sufficient to ensure that DPM-G models feature such a property. 
		Our investigation is supported by an extensive simulation study and illustrated by the analysis of an astronomical dataset consisting of physical measurements of stars in the field of the globular cluster NGC~2419.
\end{abstract}


\section{Introduction}\label{sec:intro}
	
	A natural requirement for statistical methods for density estimation and clustering is for them to be robust under affine transformations of the data. Such a desideratum is exacerbated in multivariate problems where data components are incommensurable, that is not measured in the same physical unit, and for which, thus, the definition of a metric on the sample space requires the specification of constants relating units along different axes. As an illustrative example, consider astronomical data consisting of position and velocity of stars, thus living in the so-called phase-space: a metric on such a space can be defined by setting a dimensional constant to relate positions and velocities. In this setting, any sensible statistical procedure should be robust with respect to the specification of such a constant \citep{ascasibar2005numerical,maciejewski2009phase}. 
	This is specially important considering that often scarce to no a priori guidance about dimensional constants might be available, thus making the model 
	calibration a daunting task. The motivating example of this work comes indeed from astronomy,  the dataset we consider consisting of measurements on a set of 139 stars, possibly belonging to a globular cluster called NGC~2419 \citep{Iba11}. 
	Globular clusters are sets of stars orbiting some galactic center. The NGC~2419, showed in Figure~\ref{fig:glob}, is one of the furthest known globular clusters in the Milky Way. 
	\begin{figure}[h]
		\centering
		\includegraphics[width=.9\textwidth]{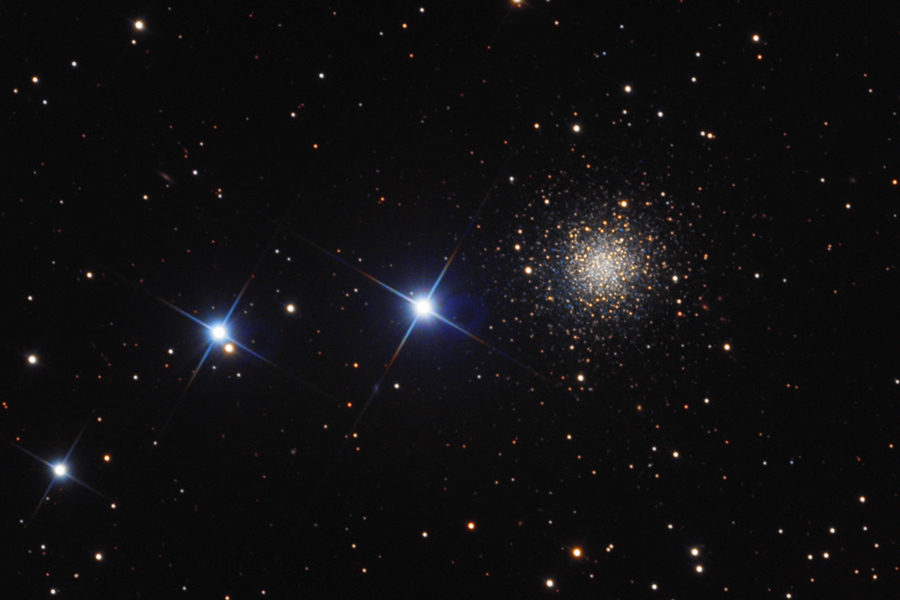}
		\caption{An image of the remote Milky Way globular cluster NGC~2419 (about 300\,000 light years away from the solar system). Picture by Bob Franke, with permission (www.bf-astro.com).} 
		\label{fig:glob}
	\end{figure}
	For each star we observe a four-dimensional vector $(Y_1,Y_2,V,[{\rm Fe/H}])$, where $(Y_1,Y_2)$ is a two-dimensional projection on the plane of the sky of the position of the star, $V$ is its line of sight velocity and $[{\rm Fe/H}]$ its metallicity, a measure of the abundance of iron relative to hydrogen. Out of these four components, only $Y_1$ and $Y_2$ are measured in the same physical unit, while dimensional constants need to be specified in order to relate position, velocity and metallicity. A key question arising with these
	data consists in identifying the stars that, among the 139 observed,
	can be rightfully considered as belonging to NGC~2419: a correct classification would be pivotal in the study of the
	globular cluster dynamics. Astronomers expect the large majority of the
	observed stars to belong to the cluster: the remaining ones, called field stars or contaminants, are
	Milky Way stars, unrelated to the cluster, that happen to appear projected  
	in the same region of the plane of the sky. In general the contaminants have different kinematic and chemical
	properties with respect to the cluster members. Considering the nature of the problem, this research question can be formalised as an unsupervised classification problem, the goal being the identification of the stars which belong to the largest cluster, which can be interpreted as the NGC~2419 globular cluster. 
	Admittedly, the terms of such a classification problem are not limited to the considered dataset but, on the contrary, are ubiquitous in astronomy and, more in general, might arise in any field where data components are incommensurable.
	
	Bayesian nonparametric methods for density estimation and clustering have been successfully applied in a wide range of fields, including genetics \citep{huelsenbeck2007inference}, bioinformatics \citep{medvedovic2002bayesian}, clinical trials \citep{xu2017decision}, econometrics \citep{otranto2002nonparametric}, to cite but a few. In this work we focus on the Dirichlet process mixture (DPM) model introduced by \citet{Lo84}, arguably the most popular Bayesian nonparametric model. Although its properties have been thoroughly studied \citep[see, e.g.,][]{hjort2010bayesian}, little attention has been dedicated to its robustness under data transformations \citep[see][]{arbel2013bayesian}. To the best of our knowledge, only \citet{bean2016transformations} and \citet{shi2018} study the effect of data transformation under a DPM model. The goal of \citet{bean2016transformations} is to transform the sample so to facilitate the estimation of univariate densities on a new scale and thus to improve the performance of the methodology; \citet{shi2018}, instead, study the consistency of DPM models under affine data transformation, when investigating the properties of the so-called low information omnibus prior for DPM models they introduce.
	
	In this paper we investigate the effect of affine transformations of the data on location-scale DPM of multivariate Gaussians (DPM-G) \citep{Mue96}, which will be introduced in Section~\ref{sec:model}. This is a very commonly used class of DPM models whose asymptotic properties have been studied by \citet{Wu10} 
	and \citet{Can17}, among others.  While rescaling the data, often for numerical convenience, is a common practice, the robustness of multivariate DPM-G models under such transformations remains essentially unaddressed to date. We fill this gap by formally studying robustness properties for a flexible specification of DPM-G models, under affine transformations of the data. Specifically, our contribution is two-fold: first, we formalise the intuitive idea that a location-scale DPM-G model on a given dataset induces a location-scale DPM-G model on rescaled data and we provide the parameters mapping for the transformed DPM-G model; second, we introduce the notion of asymptotic robustness under affine transformations of the data and show that, under mild assumptions on the true data generating process, DPM-G models feature such robustness property. As a by-product, we show that the original assumptions of  \citet{Wu10} and \citet{Can17} for ensuring posterior consistency of Dirichlet process mixtures can be simplified by removing a redundant assumption regarding the finite entropy of the model. This result, proven in Lemma~\ref{lem:improve}, can be of independent interest.
	
	Our theoretical results are supported by an extensive simulation study, focusing on both density and clustering estimation. These findings make the DPM-G model a suitable candidate to deal with problems where an informed choice of the relative scale of different dimensions seems prohibitive. We thus fit a DPM-G model to the NGC~2419 dataset and show that it provides interesting insight on the classification problem motivating this work.
	
	The rest of the paper is organised as follows. In Section~\ref{sec:model} we describe the modelling framework and introduce the notation used throughout the paper. Section~\ref{sec:theory} presents the main results of the work, with two-fold focus on finite sample properties on the one hand, and large sample asymptotics on the other. A thorough simulation study is presented in Section~\ref{sec:simulation} while Section~\ref{sec:astro} is dedicated to the analysis of the NGC~2419 dataset. Conclusions are discussed in Section~\ref{sec:conclusions}. Proofs of all results are postponed to Appendix~\ref{sec:proofs}.
	
	\section{Modelling framework}\label{sec:model}
	Let $\Xv^{(n)}:=(\Xv_1,\dots,\Xv_n)$ be a sample 
	of $d$-dimensional observations $\Xv_i:= (X_{i,1},\ldots,X_{i,d})^\intercal$ defined on some probability space $(\Omega,\mathscr{A},\mathds{P})$ and taking values in $\R^d$. 
	Consider an invertible affine transformation $g:\R^d \longrightarrow \R^d$, that is $g(\xv)=\Cm\xv+\bv$ where $\Cm$ is an invertible matrix of dimension $d\times d$ and $\bv$ 
	a $d$-dimensional column vector. The nature of the transformation $g$ is such that, if applied to a random vector $\Xv$ with probability density function $f$, it gives rise to a new random vector $g(\Xv)$ with probability density function $f_g=|\det(\Cm)|^{-1} f\circ g^{-1}$.
	
	Henceforth we denote by $\mathscr{F}$ the space of all density functions with support on $\R^d$.  
	The DPM model \citep{Lo84} 
	defines a random density taking values in $\mathscr{F}$ as 
	\begin{equation*}
	\tilde{f}(\xv) = \int_{\Theta} k(\xv;\thetav)\ddr \tilde P(\thetav),
	\end{equation*} 
	where $k(\xv;\thetav)$ is a kernel on $\R^d$ parameterized by $\thetav\in\Theta$, $\tilde P$ is a Dirichlet process (DP)  with parameters $\alpha$ (precision parameter) and $P_0:=\E[\tilde P]$  (base measure),  a distribution defined on $\Theta$  \citep{Fer73}. The almost sure discreteness of $\tilde P$ allows the random density $\tilde f$ to be rewritten as
	\begin{equation*}
	\tilde{f}(\xv) = \sum_{i=1}^{\infty} w_i k(\xv;\thetav_i),
	\end{equation*} 
	where the random atoms $\thetav_i$ are i.i.d.~from $P_0$, and the random jumps $w_i$, independent of the atoms, admit the following stick-breaking representation \citep{Set94}: given a set of random weights $v_i\simiid\text{Beta}(1,\alpha)$ (independent of the atoms $\thetav_i$), then $w_1=v_1$ and, for $j\geq 2$, $w_j=v_j \prod_{i=1}^{j-1}(1-v_i)$. While several kernels  $k(\xv;\thetav)$ have been considered in the literature, including e.g. skew-normal \citep{canale2016bayesian}, Weibull \citep{kottas2006nonparametric}, Poisson \citep{krnjajic2008parametric}, 
	here we focus on the convenient and commonly adopted Gaussian specification of  \citet{Esc95} and \citet{Mue96}. In the latter case, $k(\xv;\thetav)$ represents a $d$-dimensional Gaussian kernel $\phi_d(\xv; \muv,\Sigmam)$, provided that $\thetav=(\muv,\Sigmam)$, where the column vector $\muv$ and the matrix $\Sigmam$ represent, respectively, mean vector and covariance matrix of the Gaussian kernel. This specification defines the model referred to as $d$-dimensional location-scale Dirichlet process mixture of Gaussians (DPM-G), which can be represented in hierarchical form as 
	\begin{gather}
	\Xv_i\mid \thetav_i=(\muv_i,\Sigmam_i) \simind \phi_d(\xv_i; \muv_i,\Sigmam_i),\nonumber\\
	\thetav_i\mid \tilde P \simiid \tilde P,\label{eq:DPmmls}\\
	\tilde P \sim DP(\alpha,P_0).\nonumber
	\end{gather} 
	
	The almost sure discreteness of $\tilde P$ implies that the vector $\thetav^{(n)}:=(\thetav_1,\ldots,\thetav_n)$ might show ties with positive probability, thus leading to a partition of $\thetav^{(n)}$ into $K_n\leq n$ distinct values. This, in turn, leads to a partition of the set of observations $\Xv^{(n)}$, obtained by grouping two observations $\Xv_{i_1}$ and $\Xv_{i_2}$ together if and only if $\thetav_{i_1}=\thetav_{i_2}$. This observation implies that the posterior distribution of the random density $\tilde f$ carries useful information on the clustering structure of the data, thus making DPM-G models convenient tools for density and clustering estimation problems. 
	
	Although other specifications for the base measure can be considered \citep[see, e.g.,][]{Gor10}, we choose to work within the framework set forth by \cite{Mue96} where $P_0$ is defined as the product of two independent distributions for the location parameter $\muv$ and the scale parameter $\Sigmam$, namely a multivariate normal and an inverse-Wishart distribution, that is 
	\begin{equation}\label{eq:base}
	P_0(\ddr \muv, \ddr \Sigmam;\piv)= 
	N_d(\ddr \muv; \mv_0, \Bm_0) \times IW(\ddr \Sigmam; \nu_0, \Sm_0).
	\end{equation}
	For the sake of compactness, we use the notation $\piv:= (\mv_0, \Bm_0, \nu_0, \Sm_0)$ to denote the vector of hyperparameters characterising the base measure $P_0$. 
	We denote by $\Pi$ the prior distribution  induced on $\mathscr{F}$ by the DPM-G model \eqref{eq:DPmmls} with base measure \eqref{eq:base}.
	
	\section{Theoretical results\label{sec:theory}}
	
	\subsection{DPM-G model and affine transformations of the data}\label{sec:affine}
	
	Let $\tilde{f}_{\piv}$ be a DPM-G model defined in \eqref{eq:DPmmls}, with base measure \eqref{eq:base} and hyperparameters $\piv$. The next result shows that, for any invertible affine transformation $g(\xv)= \Cm \xv+\bv$, there exists a specification $\piv_{g}:=(\mv_0^{(g)},\Bm_0^{(g)},\nu_0^{(g)},\Sm_0^{(g)})$ of the hyperparameters characterising the base measure in \eqref{eq:base}, 
	such that the deterministic relation $\tilde f_{\piv_g}=|\det(\Cm)|^{-1}\tilde f_{\piv}\circ  g^{-1}$ holds. That is, for every $\omega\in\Omega$ and given a random vector $\Xv$ distributed according to $\tilde{f}_{\piv}(\omega)$, we have that $\tilde{f}_{\piv_g}(\omega)$ is the density of the transformed random vector $g(\Xv)$.  
	
	\begin{prop}\label{prop:exa}
		Let $\tilde{f}_{\piv}$ be a location-scale DPM-G model 
		defined as in \eqref{eq:DPmmls}, with base measure \eqref{eq:base} and hyperparameters $\piv= (\mv_0, \Bm_0, \nu_0, \Sm_0)$. For any invertible affine transformation $g(\xv)=\Cm\xv+\bv$, we have the deterministic relation
		\begin{equation*}
		\tilde f_{\piv_g}=|\det(\Cm)|^{-1}\tilde f_{\piv}\circ  g^{-1},
		\end{equation*}
		where  $\piv_{g}
		:=(\Cm \mv_0+\bv,\Cm \Bm_0 \Cm^\intercal,\nu_0,\Cm \Sm_0 \Cm^\intercal)$.
	\end{prop}
	
	While Proposition~\ref{prop:exa} can be derived from general properties of the Dirichlet process \citep[see][]{lijoi2009distributional}, a direct proof is provided in Appendix~\ref{subsec:pro1}. This result implies that, for any invertible affine transformation $g$, modelling the set of observations $\Xv^{(n)}$ with a DPM-G model \eqref{eq:DPmmls}, with base measure \eqref{eq:base} and hyperparameters $\piv$, is equivalent with assuming the same model 
	with transformed hyperparameters $\piv_g$, for the transformed observations $g(\Xv)^{(n)}:=(g(\Xv_1),\ldots,g(\Xv_n))$. As a by-product, the same posterior inference can be drawn conditionally on both the original and the transformed set of observations, as the conditional distribution of the random density $\tilde f_{\piv_g}$, given $g(\Xv)^{(n)}$, coincides with the conditional distribution of $|\det(\Cm)|^{-1}\tilde{f}_{\piv} \circ g^{-1} $, given $\Xv^{(n)}$. Proposition~\ref{prop:exa} thus provides a formal justification for the procedure of transforming data, e.g. via standardisation or normalisation, often adopted to achieve numerical efficiency: as long as the prior specification of the hyperparameters of a DPM-G model respects the condition of Proposition~\ref{prop:exa}, 
	transforming the data does not affect posterior inference.
	
	\subsubsection{Empirical Bayes approach}
	
	The elicitation of an honest prior, thus independent of the data, for the hyperparameters $\piv$ of the base measure \eqref{eq:base} of a DPM model is in general a difficult task. A popular practice, therefore, consists in setting the hyperparameters equal to some empirical estimates $\hat{\piv}(\Xv^{(n)})$, by applying the so-called  empirical Bayes approach \citep[see, e.g.,][]{Leh06}. Recent investigations \citep{Pet14,Don18} provide a theoretical justification of this hybrid procedure by shedding light on its asymptotic properties. We show here that this procedure satisfies the assumptions of Proposition~\ref{prop:exa} and, thus, guarantees that posterior Bayesian inference, under an empirical Bayes approach, is not affected by affine transformations to the data.\\

	A commonly used empirical Bayes approach for specifying the hyperparameters $\piv$ of a DPM-G model, defined as in \eqref{eq:DPmmls} and \eqref{eq:base}, consists in setting
		\begin{equation}\label{eq:EB}
		\mv_0=\overline{\Xv},\qquad\quad \Bm_0=\frac{1}{\gamma_1}\Sm_\Xv^2,\qquad\quad \Sm_0= \frac{\nu_0 - d - 1}{\gamma_2} \Sm_\Xv^2,
		\end{equation}
		where $\overline{\Xv}=\sum_{i=1}^n \Xv_i/n$ and $\Sm_\Xv^2=\sum_{i=1}^n(\Xv_i-\overline{\Xv})(\Xv_i-\overline{\Xv})^\intercal/(n-1)$ are the sample mean vector and the sample covariance matrix, respectively, and $\gamma_1,\gamma_2>0$, $\nu_0>d+1$. 
		This specification for the hyperparameters $\piv$ has a straightforward interpretation. Namely, the parameter $\mv_0$, mean of the prior guess distribution of $\muv$, can be interpreted as the overall mean value and, in absence of available prior information, set equal to the observed sample mean. Similarly, the parameter $\Bm_0$, covariance matrix of the prior guess distribution of $\muv$, is set equal to a penalised version of the sample covariance matrix $\Sm^2_\Xv$, where $\gamma_1$ takes on the interpretation of the size of the ideal prior sample upon which the prior guess on the distribution of $\muv$ is based. Similarly, the hyperparameter $\Sm_0$ is set equal to a penalised version of the sample covariance matrix $\Sm_\Xv^2$, choice that corresponds to the prior guess that the covariance matrix of each component of the mixture coincides with a rescaled version of the sample covariance matrix. Specifically, $\Sm_0= \Sm^2_\Xv (\nu_0 - d - 1)/\gamma_2$ follows by setting $\E[\Sigmam]=\Sm_\Xv^2/\gamma_2$ and observing that, by standard properties of the inverse-Wishart distribution, 
		$\E[\Sigmam]=\Sm_0/(\nu_0 - d - 1)$. Finally the parameter $\nu_0$ 
		takes on the interpretation of the size of an ideal prior sample upon which the prior guess $\Sm_0$ is based. Next we focus on the setting of the hyperparameters $\piv_g$, given the transformed observations $g(\Xv)^{(n)}$. The same empirical Bayes procedure adopted in \eqref{eq:EB} leads to
		\begin{equation*}\label{eq:EB2}
		\mv_0^{(g)}=\overline{g(\Xv)}=\Cm \mv_0 + \bv,\qquad \Bm^{(g)}_0=\frac{1}{\gamma_1}\Sm_{g(\Xv)}^2,\qquad \Sm^{(g)}_0= \frac{\nu_0 - d - 1}{\gamma_2} \Sm_{g(\Xv)}^2.
		\end{equation*}
		Observing that $\Sm_{g(\Xv)}^2=\Cm \Sm_\Xv^2 \Cm^\intercal$ and setting $\nu_0^{(g)}=\nu_0$ shows that the described empirical Bayes procedure corresponds to $\piv_g=(\Cm \mv_0+\bv,\Cm \Bm_0 \Cm^\intercal,\nu_0,\Cm \Sm_0 \Cm^\intercal)$ and, thus, by Proposition~\ref{prop:exa}, $\tilde f_{\piv_g}=|\det(\Cm)|^{-1}\tilde f_{\piv}\circ  g^{-1}$.
	
	\subsection{Large $n$ asymptotic robustness}\label{sec:asymptotics}
	
	We investigate the effect of affine transformations of the data on DPM-G models by studying the asymptotic behaviour of the resulting posterior distribution in the large sample size regime. To this end, we fit the same DPM-G model $\tilde{f}_{\piv}$, defined in \eqref{eq:DPmmls} and \eqref{eq:base}, to two versions of the data, that is $\Xv^{(n)}$ and $g(\Xv)^{(n)}$, by using the exact same specification for the hyperparameters $\piv$. Under this setting, the assumptions of Proposition~\ref{prop:exa} are not met and the posterior distributions obtained by conditioning on the two sets of observations are different random distributions which, thus, might lead to different statistical conclusions. The main result of this section shows that, under mild conditions on the true generating distribution of the observations, the posterior distributions obtained by conditioning $\tilde{f}_{\piv}$ on the two sets of observations $\Xv^{(n)}$ and $g(\Xv)^{(n)}$, become more and more similar, up to an affine reparametrisation, as the sample size $n$ grows. More specifically we show that the probability mass of the joint distribution of these two conditional random densities concentrates in a neighbourhood of $\{(f_1,f_2)\in\mathcal{F}\times\mathcal{F}\text{ s.t. }f_1=|\det(\Cm)| f_2 \circ g\}$ as $n$ goes to infinity. Henceforth we will say that the DPM-G model \eqref{eq:DPmmls} with base measure \eqref{eq:base} is asymptotically robust to affine transformation of the data. The rest of the section formalises and discusses this result. 
	We consider a metric $\rho$ on $\mathscr{F}$ which can be equivalently defined as the Hellinger distance $\rho(f_1,f_2)=\{\int (\sqrt{f_1(\xv)}-\sqrt{f_2(\xv)})^2\ddr \xv \}^{1/2}$ or the $L^1$ distance $\rho(f_1,f_2)=\int |f_1(\xv)-f_2(\xv))|\ddr \xv$ between densities $f_1$ and $f_2$ in $\mathscr{F}$, and we denote by $\|\cdot\|$ the Euclidean norm on $\R^d$. Moreover, we adopt here the usual frequentist validation approach in the large $n$ regime, working `as if' the observations $\Xv^{(n)}$ were generated from a true and fixed data generating process $F^*$ \citep[see for instance][]{rousseau2016frequentist}. We introduce the notation $F_{n}^*$ to denote the $n$-fold product measure $F^*\times \cdots \times F^*$, and we assume that $F^*$ admits a density function with respect to the Lebesgue measure, denoted by  $f^*$.	In the setting we consider, the same model $\tilde{f}_{\piv}$ defined in \eqref{eq:DPmmls} and \eqref{eq:base} is fitted to $\Xv^{(n)}$ and $g(\Xv)^{(n)}$, thus leading to two distinct posterior random densities, with distributions on $\mathscr{F}$ denoted by $\Pi(\,\cdot\, \mid \Xv^{(n)})$ and $\Pi(\,\cdot\, \mid g(\Xv)^{(n)})$, respectively. 
	We use the notation $\Pi_2(\cdot\mid \Xv^{(n)})$ to refer to their joint posterior distribution on $\mathscr{F}\times\mathscr{F}$.
	
	\begin{theo}\label{thm:asym}
		Let $f^*\in\mathscr{F}$, true generating density of $\Xv^{(n)}$, satisfy the conditions
		\begin{enumerate}
			\item[{\normalfont A1}.] $0 < f^*(\xv) < M$, for some constant $M$ and for all $\xv \in \R^d$,
			\item[{\normalfont A2}.] for some $\eta > 0$, $\int \|\xv\|^{2(1+\eta)}f^*(\xv)\ddr\xv < \infty$,
			\item[{\normalfont A3}.] $\xv\mapsto f^*(\xv)\log^2(\varphi_\delta (\xv))$ is bounded on $\R^d$, where $\varphi_\delta(\xv) = \inf_{\{\tv\,:\,\|\tv-\xv\|<\delta\}}f^*(\tv)$.
		\end{enumerate}
		Let $g:\R^d \longrightarrow \R^d$ be an invertible affine transformation and $\tilde{f}_{\piv}$ be the random density 
		induced by a DPM-G as \eqref{eq:DPmmls} with base measure \eqref{eq:base} where $\nu_0>(d + 1)(2d - 3)$. 
		Then, for any $\varepsilon>0$,
		\begin{equation*}
		\Pi_2((f_1,f_2): \rho(f_1,|\det(\Cm)| f_2 \circ g)<\varepsilon \mid \Xv^{(n)})\longrightarrow 1
		\end{equation*}
		in $F_{n}^*$-probability, as $n\to\infty$.
	\end{theo}
	
	It is worth stressing that, while in line with the usual posterior consistency approach the existence of a true data generating process $F^*$ is postulated, the focus of Theorem~\ref{thm:asym} is not on the asymptotic behaviour of the posterior distribution with respect to the true data generating process but rather on the relative behaviour of two posterior distributions, obtained by conditioning the same model on two sets of observations which coincide up to an affine transformation. More specifically, according to Theorem~\ref{thm:asym}, when the sample size grows, the joint distribution $\Pi_2(\cdot\mid \Xv^{(n)})$ concentrates its mass on a subset of the space $\mathscr{F}\times\mathscr{F}$ where the distance $\rho$ between $f_1$ and $|\det(\Cm)| f_2 \circ g$ is smaller than $\varepsilon$. In other terms, the two posterior distributions get similar, up to the affine transformation, as $n$ becomes large.
	
	The assumptions of Theorem~\ref{thm:asym}  refer to the true generating distribution $f^*$ of $\Xv^{(n)}$. 
	Assumption A1 requires $f^*$ to be bounded and fully supported on $\R^d$.
	Assumption A2 requires the tails of $f^*$ to be thin enough for some moment of order strictly larger than two to exist. Such an assumption is not met, for example, by a Student's $t$-distribution with two degrees of freedom, case which will be considered in the simulation study of Section~\ref{sec:simulation}. Finally, assumption A3 is a weak condition ensuring local regularity of the entropy of $f^*$.
	
	The proof of Theorem~\ref{thm:asym} is based on previous results proved by \cite{Wu08} and \cite{Can17} in order to derive the so-called Kullback--Leibler property at $f^*$  for some mixtures of Gaussians models. Importantly, in Lemma~\ref{lem:improve} (see Appendix~\ref{thm:asym}), we improve upon their results by showing that the set of assumptions required by \cite{Wu08} and \cite{Can17} can be reduced to the simpler set of assumptions A1, A2 and A3 of Theorem~\ref{thm:asym} by removing a redundant assumption. More specifically, we prove 
	that A1, A2 and A3  imply that $f^*$ has finite entropy and  regular local entropy, conditions required in the aforementioned works.

	\section{Simulation study}\label{sec:simulation}

	We ran an extensive simulation study with a two-fold goal: 1) providing empirical support to our result on the large $n$ asymptotic robustness of a DPM-G model, 
	under affine transformations of the data; 2) investigating whether an analogous robustness property holds when DPM-G models are adopted to make inference on the clustering structure of the data. To this end, we considered two distinct data-generating distributions, which allowed us to highlight different facets of DPM-G models. In the first case, data are generated from a mixture of bivariate Gaussians, distribution which satisfies the conditions of Theorem~\ref{thm:asym}. 
	This study complements our asymptotic result with a numerical investigation of the finite $n$ behaviour of DPM-G models, when data undergo an affine transformation. Moreover, the same data are used to perform a numerical study on the effect of data transformation and sample size on the number of clusters on the estimated partition. While not directly related to theoretical results of Section~\ref{sec:theory}, this part of the study is relevant in view of the astronomical application of Section~\ref{sec:astro} where a DPM-G model will be used for unsupervised clustering. The second scenario we considered does not satisfy the set of assumptions of Theorem~\ref{thm:asym}, as data are generated from univariate Student's $t$-distribution with two degrees of freedom, thus breaking assumption A2. Our study, in this case, aims at assessing the robustness of DPM-G models when the sufficient conditions of Theorem~\ref{thm:asym} are not met.

	\subsection{Data from mixture of Gaussians}\label{subsec:mix}
	
	The first part of the simulation study focuses on the analysis of data generated from a mixture of Gaussians.
	Specifically, we considered three sample sizes, namely $n = 100$, $n = 300$ and $n = 1\,000$, and we generated 100 samples $\Xm^{(n)}$, for each $n$, from a mixture of two Gaussian components with  density function
		\begin{equation*}
	f(\xv) = \frac{1}{2} \phi_2\left(\xv; \mv_1, \Sm_1\right) + \frac{1}{2} \phi_2\left(\xv;\mv_2,\Sm_2\right),
	\end{equation*}
	 where $\phi_2(\cdot;\mv,\Sm)$ denotes the density function of a two-dimensional Gaussian distribution with mean vector $\mv$ and covariance matrix $\Sm$, and the two components of the mixture are characterized by the parameters 
	 \begin{equation*}
	     \mv_1=(-2,-2),\quad\Sm_1=\begin{bmatrix} 1 & 0.85 \\ 0.85 & 1 \end{bmatrix},\quad\mv_2=(2,2),\quad \Sm_2=\begin{bmatrix} 1 & 0 \\ 0 & 1 \end{bmatrix}.\end{equation*}
	     
\noindent In order to test the robustness of the model under affine transformations of the data, we compressed or stretched the generated datasets by using five different constants, namely $c = 1/5$, $c = 1/2$, $c = 1$, $c = 2$ and $c = 5$. For each constant, we multiplied the simulated data by $c$, thus
	obtaining a transformed dataset $\Xm_c^{(n)} := c \Xm^{(n)}$. We then fitted a DPM-G model, specified as in \eqref{eq:DPmmls} and \eqref{eq:base}, to each one of the $5\times 3\times 100 = 1\,500$ resulting datasets. In order to enhance the flexibility of the model, we completed its specification by setting a normal/inverse-Wishart prior distribution for the hyperparameters $(\mv_0,\Bm_0)$ of the base measure \eqref{eq:base}. 
	Namely, we set 
	$\Bm_0 \sim IW(4, \diag(\mathbf{15}))$ and $\mv_0\mid\Bm_0 \sim N(0, \Bm_0)$, specification chosen so that $\E[\muv] = \mathbf{0}$ and to guarantee a prior guess on the location component $\muv$ flat enough to cover the support of the non-transformed data. As for the scale component of the base measure \eqref{eq:base}, we set $(\nu_0,\Sm_0)=(4, \diag(\mathbf 1))$. Finally, the precision parameter $\alpha$ of the Dirichlet process was set equal to 1.
	
	Realisations of the mean of the posterior distribution were obtained by means of a Gibbs sampler relying on a Blackwell--McQueen P\'olya urn scheme \citep[see][]{Mue96}, implemented in the \href{https://github.com/rcorradin/AFFINEpack}{\textsf{AFFINEpack}} \textsf{R} package\footnote{\label{foot:package}The package is available at \href{https://github.com/rcorradin/AFFINEpack}{https://github.com/rcorradin/AFFINEpack} and can be installed via \textsf{devtools}. For reproducibility,  the code is available at  \href{https://github.com/rcorradin/Affine}{https://github.com/rcorradin/Affine}.}.
	For each replicate, posterior inference was drawn based on $5\,000$ iterations, obtained after discarding the first $2\,500$. Convergence of the chains was assessed by visually investigating the traceplots of some randomly selected replicates, which did not provide indication against it.
	
	\begin{figure}[h]
		\centering
		\includegraphics[width=0.99\textwidth]{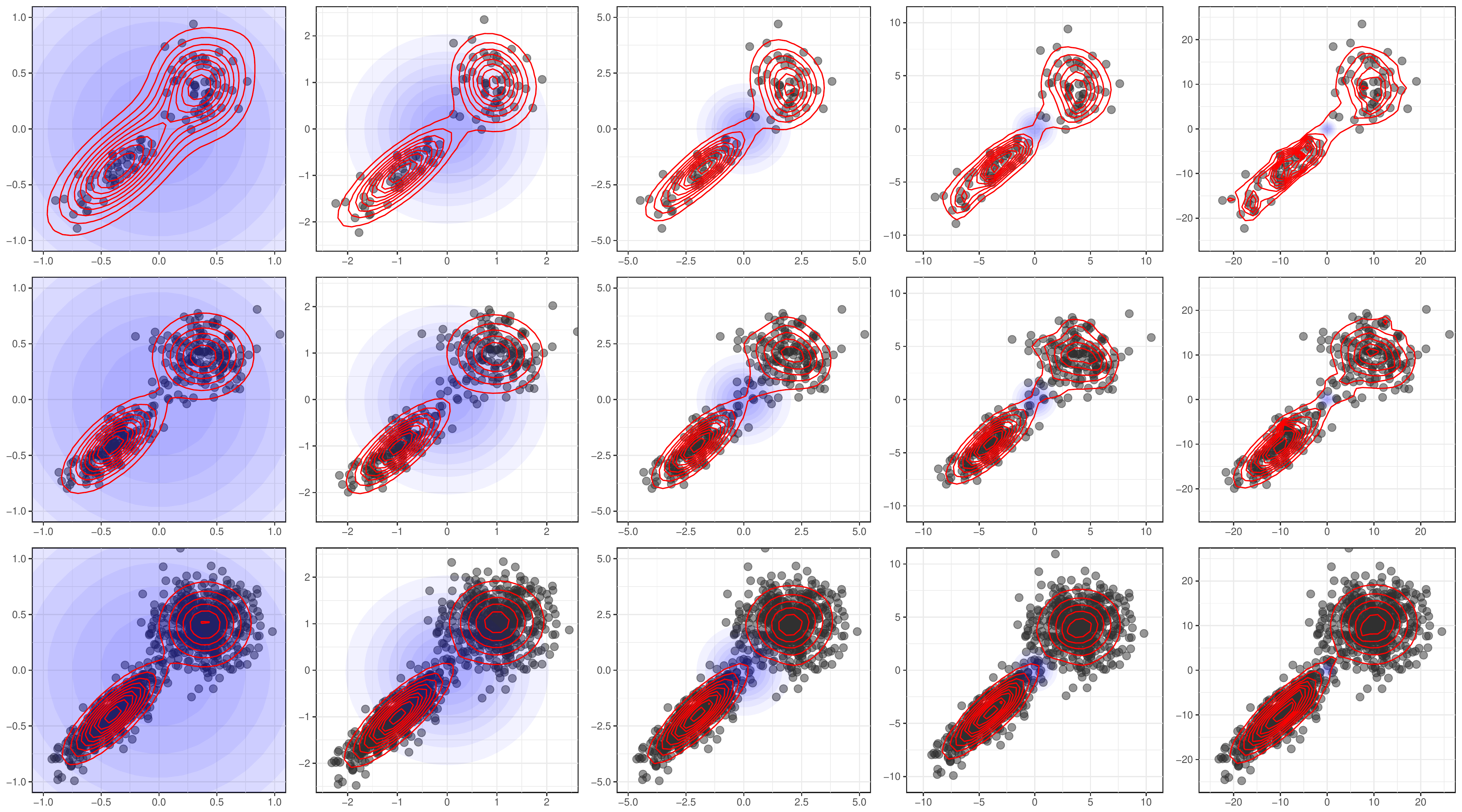} 
		\caption{Simulation study, data generated from a mixture of Gaussians. Based on a single replicate of the samples $\Xv^{(100)}$, $\Xv^{(300)}$ and $\Xv^{(1\,000)}$, scatter plots of the data (grey dots), contour plots of the estimated densities based on a DPM-G model (red curves) and contour plots for the expected prior density (blue filled curves). Left to right: rescaling constant $c=1/5$, $c=1/2$, $c=1$, $c=2$, $c=5$. Top to bottom: sample size $n = 100$, $n = 300$, $n = 1\,000$.}
		\label{fig:splotsim}
	\end{figure}
	
	Figure~\ref{fig:splotsim} shows, for every $n\in\{100,300,1\,000\}$ and $c\in\{1/5,1/2,1,2,5\}$, a contour plot of the estimated posterior densities. 
	The difference between estimated densities, across different values of $c$, is apparent when $n=100$, with the two extreme cases, namely $c=1/5$ and $c=5$, displaying very different contour lines and possibly suggesting a different number of modes in the estimated density. For larger sample sizes, this difference is less evident and, when $n=1\,000$, the contour plots are hardly distinguishable. These qualitative observations are in agreement with the large $n$ asymptotic results of Theorem~\ref{thm:asym}. The plots of Figure~\ref{fig:splotsim} refer to a single realisation of the samples $\Xv^{(100)}$, $\Xv^{(300)}$ and $\Xv^{(1\,000)}$ considered in the simulation study, although qualitatively similar results can be found in almost any replicate.
	
	The findings drawn from a visual inspection of Figure~\ref{fig:splotsim} were confirmed by assessing the distance between estimated posterior densities. Specifically, for any considered sample size $n$ and for any pair of values $c_1$ and $c_2$ taken by the constant $c$, we approximately evaluated the $L^1$ distance between the suitably rescaled estimated posterior densities obtained conditionally on $\Xv_{c_1}^{(n)}$ and on $\Xv_{c_2}^{(n)}$. The results of such analysis are shown in Figure~\ref{fig:L1plot} and indicate that, as the sample size grows, the difference in terms of $L^1$ distance strictly decreases. 
	
	\begin{figure}[h]
		\centering
		\includegraphics[width=0.99\textwidth]{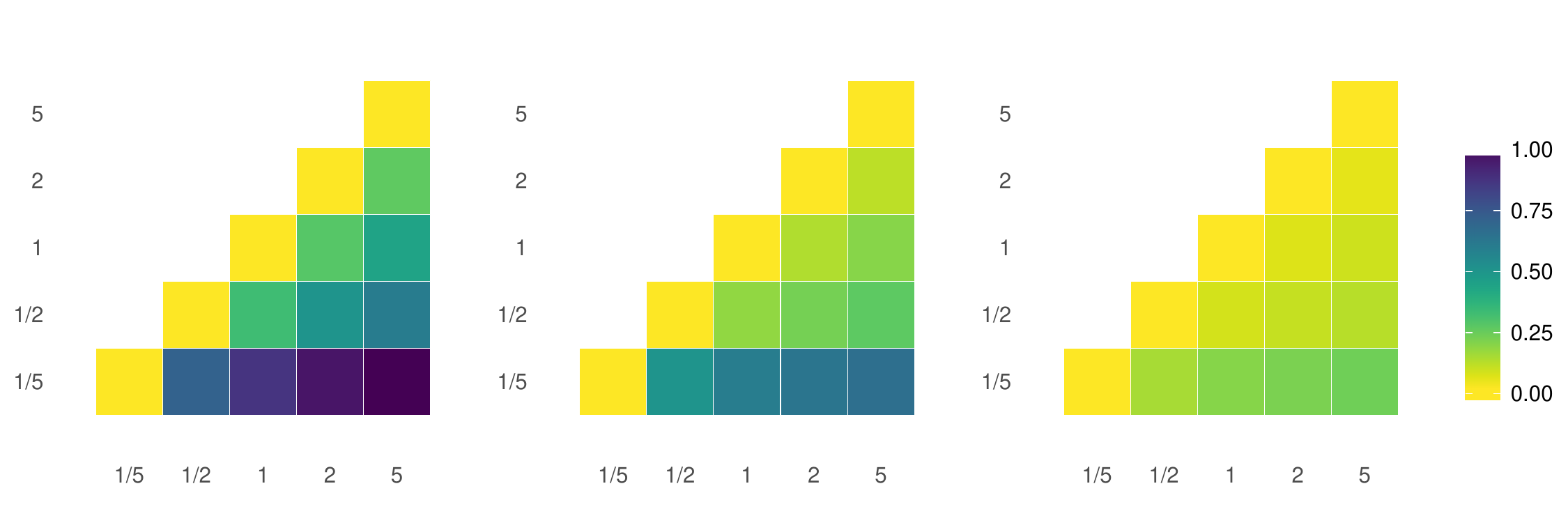} 
		\caption{Simulation study, data generated from a mixture of Gaussians. Normalised $L^1$ distances (all distances are divided by the largest observed distance) between suitably rescaled estimated densities, conditionally on data rescaled by means of different constants $c_1$ ($X$ axis) and $c_2$ ($Y$ axis),  where $c_1$ and $c_2$ denote the scaling factors used to transform the data,
			averaged over $100$ replications. Left to right: sample size $n=100$, sample size $n=300$, sample size $n=1\,000$.}
		\label{fig:L1plot}
	\end{figure}
	
	The posterior distribution of the random density induced by a DPM-G model provides interesting insight also on the clustering structure of the data. The second goal of the simulation study, thus, consisted in investigating the impact of the scaling factor $c$ on the estimated number of groups in the partition induced on the data.
	To this end, for each considered $n$ and $c$, we estimated
	$\hat K_n^{(\text{VI})}$, the number of groups in the optimal partition estimated using a procedure introduced by \cite{Wad17} and based on the variation of information loss function. In light of known inconsistency results for the posterior distribution of the number of components under a DPM-G model \citep[see, for instance,][]{miller2013simple}, the numerical findings of this part of the simulation study contribute to shed some light on the large $n$ behaviour of $\hat{K}_n^{\text{(VI)}}$.
	The average values for this quantity, over 100 replicates, are reported in Table~\ref{tab:simresult}. 
	\begin{table}[h]
		\caption{Simulation study, data generated from a mixture of Gaussians. Averages over 100 replicates for $\hat K_n^{(\text{VI})}$, the number of clusters of the optimal partition estimated by means of \citet{Wad17}'s variation of information method. Left to right: rescaling constant $c=1/5$, $c=1/2$, $c=1$, $c=2$, $c=5$. Top to bottom: sample size $n = 100$, $n = 300$, $n = 1\,000$.}
		\label{tab:simresult}
		\centering
		\begin{tabular}{clllll}
			\hline
			\multicolumn{1}{l}{} & $c = 1/5$ & $c = 1/2$ & $c = 1$ & $c = 2$ & $c = 5$ \\ \hline
			\multicolumn{1}{l}{$n=100$}   &1.81 &2.04 &2.84 &5.96 &10.52    \\ 
			\multicolumn{1}{l}{$n=300$}   &2.00 &2.03 &2.20 &2.82  &5.18    \\ 
			\multicolumn{1}{l}{$n=1\,000$}  &2.00 &2.00 &2.04 &2.05  &2.12    \\ \hline
		\end{tabular}
	\end{table}
	There appears to be a clear trend suggesting that a larger scaling constant $c$ leads to a larger $\hat K_n^{(\text{VI})}$: this finding is consistent with the fact that, if the data are stretched while the prior specification is kept unchanged, then we expect the estimated posterior density to need a larger number of Gaussian components to cover the support of the sample. 
	For the purpose of this simulation study the main quantity of interest is the ratio between the estimated number of groups under any two distinct values $c_1$ and $c_2$ for the scaling constant $c$, that is
	$\hat K_{n,c_1}^{(\text{VI})}/\hat K_{n,c_2}^{(\text{VI})}$. The results presented in Table~\ref{tab:simresult} clearly indicate that, as the sample size $n$ becomes large, such ratios tend to approach $1$. This suggests that the large $n$ robustness property of the DPM-G model nicely translates to an equivalent notion of robustness in terms of the estimated number of groups $\hat K_{n}^{(\text{VI})}$ in the data.
	
	\subsection{Data from Student's $t$-distribution}\label{sec:misst}
	
	The second part of the simulation study deals with the same simulation scenarios ($n\in\{100,300,1\,000\}$ and $c=\{1/5,1/2,1,2,5\}$) considered in Section~\ref{subsec:mix}, with the difference that data are generated from a Student's $t$-distribution with two degrees of freedom. It is important to stress that such a distribution does not have finite variance and therefore does not meet assumption A2 of Theorem~\ref{thm:asym}. Also in this case we considered 100 replicates for each considered simulation scenario. 
	
	We analysed each dataset with a univariate version of the DPM-G model specified in \eqref{eq:DPmmls} and \eqref{eq:base}. That is, we considered a univariate Gaussian kernel and a base measure defined as the product of two independent distributions, a univariate normal distribution for the location parameter $\mu\sim N(m_0, s_0^2)$ and an inverse-gamma distribution for the scale parameter $\sigma^2\sim IG(a_0,b_0)$. The model specification is completed by setting 
	$a_0 = 2$ and $b_0 = 1$, so that $\E[\sigma^2] = 1$, and by considering a normal/inverse-gamma distribution for the hyperparameters $(m_0, s_0^2)$, specifically $s_0^2 \sim IG(2,1)$ and $m_0 \mid s_0^2 \sim N(0, s_0^2)$. Finally, the precision parameter $\alpha$ of the Dirichlet process was set equal to 1. Realisations of the mean of the posterior distribution were obtained by means of a Gibbs sampler relying on a Blackwell--McQueen P\'olya urn scheme\footnote{See footnote~\ref{foot:package}.}. Posterior inference was drawn based on $5\,000$ iterations, after a burn-in period of $2\,500$ iterations. 
	We assessed the convergence of the chains by visually investigating traceplots, which did not provide indication against it. 
	
	Also for these data, for any considered sample size $n$ and for any pair of values $c_1$ and $c_2$ taken by the constant $c$, we approximately evaluated the $L^1$ distance between the suitably rescaled estimated posterior densities obtained conditionally on $\Xv_{c_1}^{(n)}$ and on $\Xv_{c_2}^{(n)}$. The results of such analysis are displayed in Figure~\ref{fig:L1plot_tstudent} and indicate that, as the sample size grows, the $L^1$ distance decreases. This qualitative findings suggest that asymptotic robustness might hold also for data generated from a distribution not meeting the assumptions of Theorem~\ref{thm:asym}.

	\begin{figure}[h]
		\centering
		\includegraphics[width=0.99\textwidth]{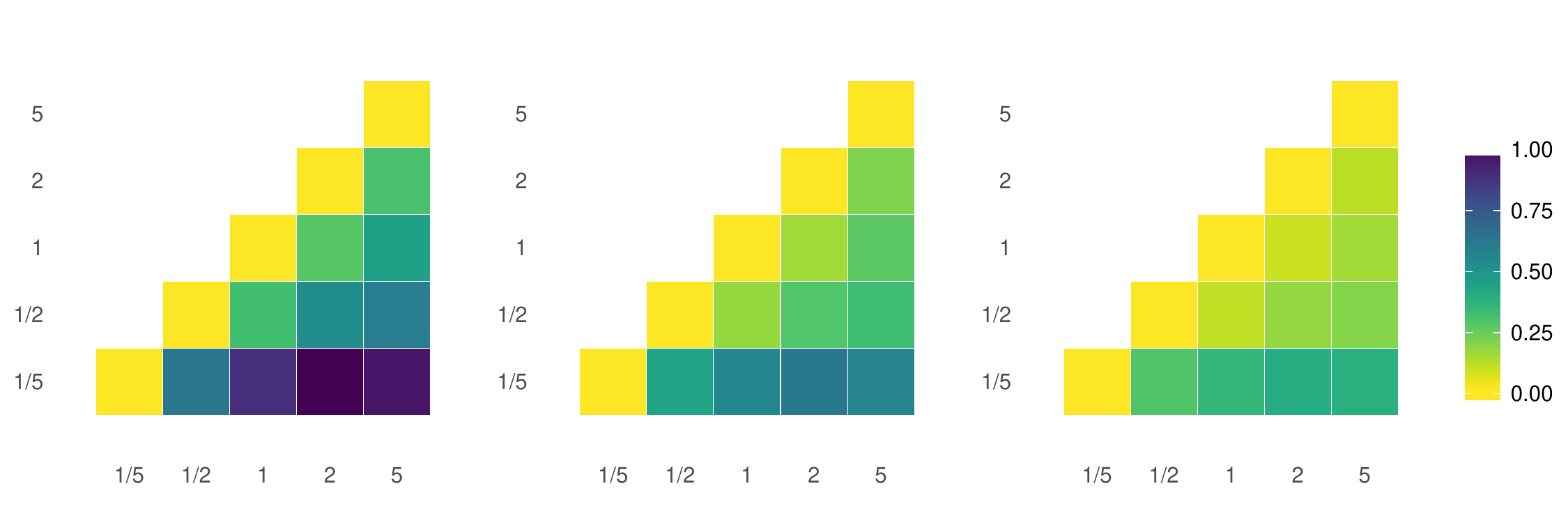} 
		\caption{Simulation study, data generated from a Student's $t$-distribution. Normalised $L^1$ distances (all distances are divided by the largest observed distance) between suitably rescaled estimated densities, conditionally on data rescaled by means of different constants $c_1$ ($X$ axis) and $c_2$ ($Y$ axis), where $c_1$ and $c_2$ denote the scaling factors used to transform the data,	 averaged over $100$ replications. Left to right: sample size $n=100$, sample size $n=300$, sample size $n=1\,000$.}
		\label{fig:L1plot_tstudent}
	\end{figure}

	\section{Astronomical data}\label{sec:astro}
	
	The large $n$ asymptotic robustness to affine transformation of the DPM-G model makes it a suitable candidate also for analysing data whose components are not commensurable and for which an informed choice of the relative scale of different dimensions seems prohibitive. We fitted the DPM-G model, specified as in \eqref{eq:DPmmls} and with base measure \eqref{eq:base}, to the NGC~2419 dataset described in Section~\ref{sec:intro}.
	The ultimate goal of our analysis consists in classifying stars as belonging to the NGC~2419 globular cluster or as being contaminants: an accurate classification is crucial for the astronomers to study the dynamics of the globular cluster. Since the large majority of the stars in the dataset is expected to belong to the globular cluster, with only a few of them being contaminants, we will identify the globular cluster as the largest group in the estimated partition of the dataset.
	
	Prior to any analysis, data were standardised component by component, the legitimacy of such procedure following from the robustness results of Theorem~\ref{thm:asym}.
	Hyperprior distributions were specified for the location parameter of the base measure \eqref{eq:base} and on the DP precision parameter $\alpha$. Specifically, 
	$\Bm_0 \sim IW(6, \diag(\mathbf{15}))$ and $\mv_0\mid\Bm_0 \sim N(0, \Bm_0)$, specification chosen to guarantee a prior guess on the location component $\muv$ flat enough to cover the support of the data and centered at $\mathbf 0$. In addition, $\alpha$ was given a gamma prior distribution with unit shape parameter and rate parameter equal to $5.26$, so that, a priori, $\alpha_0:=\E[\alpha] \simeq 0.19$. This leads to an expected number of components $K_n$ in a sample of size $n=139$ from a DP equal to $\sum_{i=1}^{n}\alpha_0/(\alpha_0 + i - 1) \simeq 2$, thus reflecting the prior opinion of astronomers who would expect two distinct groups of stars in the dataset. Finally, as far as the scale component of the base measure \eqref{eq:base} is concerned, we set $(\nu_0,\Sm_0)=(26,\diag(\mathbf{21}))$, where the number of degrees of freedom $\nu_0=26$ of the inverse-Wishart distribution was chosen so that to satisfy the conditions of Theorem~\ref{thm:asym} and, in turn, the scale matrix $\Sm_0=\diag(\mathbf{21})$ so that $\E[\Sigmam] = \diag (\mathbf 1)$. 
	Realisations of the mean of the posterior distribution were obtained by means of a Gibbs sampler relying on a Blackwell--McQueen P\'olya urn scheme\footnote{See footnote~\ref{foot:package}.}. In turn, posterior inference was drawn based on $20\,000$ iterations, after a burn-in period of $5\,000$ iterations. Convergence of the chains was assessed by visually investigating traceplots, which did not provide indication against it. 

    \begin{figure}[h]
        \centering
        \includegraphics[width=0.99\textwidth]{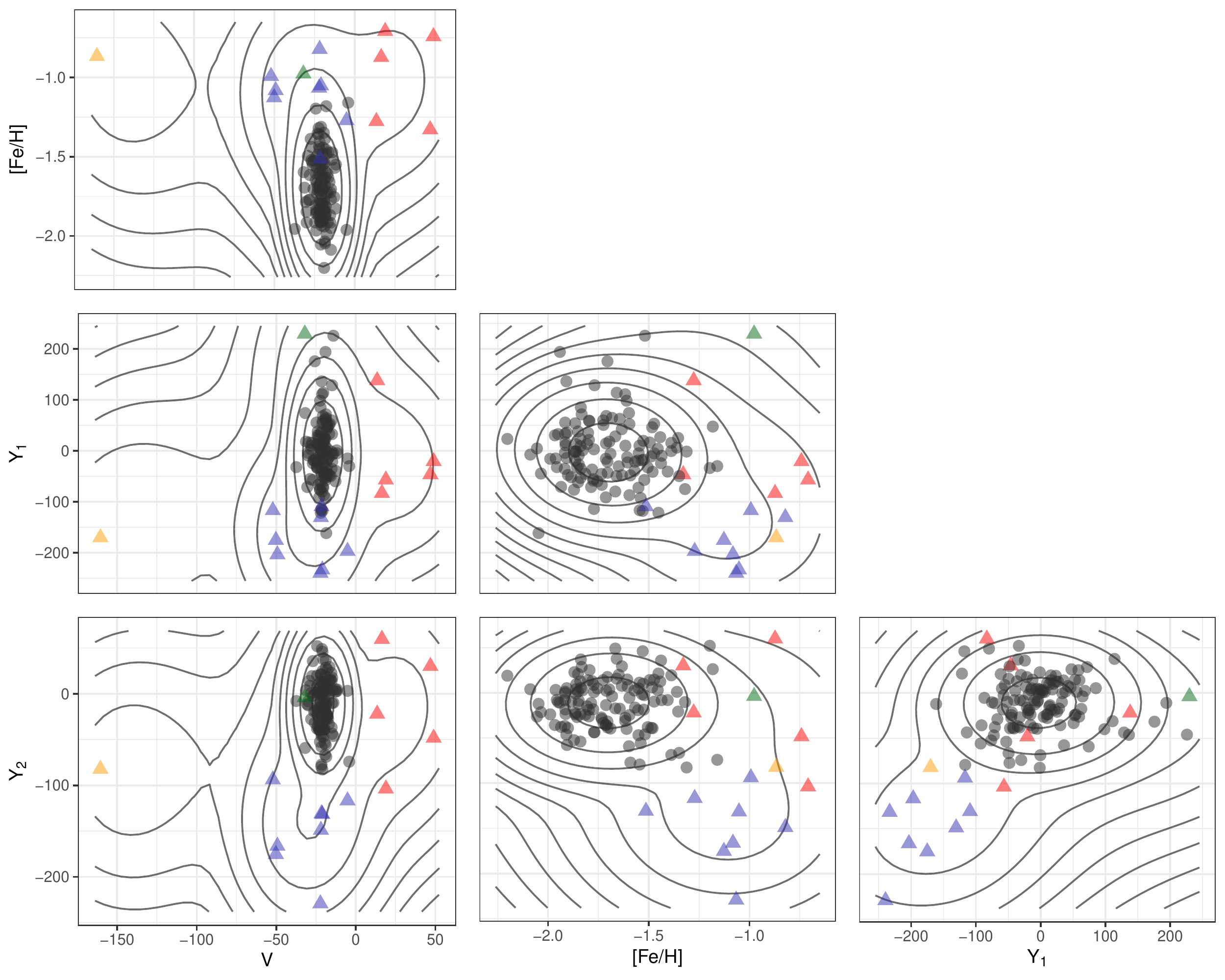} 
        \caption{NGC~2419 data. Contour plots of the bivariate log marginal densities estimated via DPM-G model (log densities are used for better visualization). Partition estimated via DPM-G model combined with \citet{Wad17}'s variation of information method. Five groups are detected: the largest group (grey dots), group A (blue triangles), group B (red triangles), group C (one orange triangle), group D (one green triangle).}
        \label{fig:dens}
    \end{figure}
    		
Figure~\ref{fig:dens} displays contour plots for the six two-dimensional projections of the estimated posterior density, with the scatter plots of the dataset with individual observations coloured according to their membership in the optimal partition estimated via the variation of information method of \citet{Wad17}, and labeled as main group (grey circles) and other groups (coloured triangles). 
	The estimated partition is composed of five groups. The largest one, identified as the globular cluster, consists of 124 stars. The remaining 15 stars are thus considered contaminants and are further divided into four groups, one composed by eight stars (group A), 
	one containing five stars (group B) 
	and two singletons (groups C and D). 
	A visual investigation of Figure~\ref{fig:dens} suggests that stars in group A differ from those in the globular cluster in terms of metallicity and position, with the contaminants characterised by larger values for $[{\rm Fe/H}]$ and smaller values for $Y_1$ and $Y_2$. Stars in group B differ from the globular cluster in terms of velocity and metallicity, with the contaminants showing larger values for $V$ and $[{\rm Fe/H}]$. Finally, groups C and D are singletons, the first one being characterised by a high metallicity and an extremely small value for the velocity, the second one showing large values for both metallicity and location $Y_1$. Our unsupervised statistical clustering can be compared to the clustering of \citet{Iba11} (described in their Figure 4) obtained by means of ad hoc physical considerations. Specifically, once the best fitting physical model, in the class of either Newtonian or Modified Newtonian Dynamics models, is detected, they use it in order to compute the average values of the physical variables describing the stars. Stars are then assigned to the globular cluster based on a comparison between their velocity and the average model velocity: those lying close enough are deemed to belong to the cluster, while the others are considered as potential contaminants. For the latter, the evidence of being contaminants is measured by evaluating how distant their metallicity is from the average model one. Two classifications are then proposed: the first one assigns to the globular cluster only the 118 stars for which the evidence seems strong, the second and less conservative strategy classifies as belonging to the globular cluster a total of 130 stars. Following this distinction and for the sake of simplicity, we summarise the results of \citet{Iba11}'s analysis, by devising three groups of stars:
	\begin{itemize}
		\item[-] \emph{globular cluster}: 118 stars deemed to belong to the globular cluster,
		\item[-] \textit{likely globular cluster}: 12 stars assigned to the globular cluster only when the less conservative 
		procedure is adopted,
		\item[-] \textit{contaminants}: 9 stars with strong evidence of being contaminants.
	\end{itemize} 
	\renewcommand{\arraystretch}{1.5}
	\begin{table}[h]
		\caption{NGC~2419 data. Comparison between the groups identified by \citet{Iba11} and the groups estimated via DPM-G model.}
		\label{tab:VSibata}
		\centering
		\begin{tabular}{llc*{5}{c}}
			\toprule
			&       && \multicolumn{5}{c}{DPM-G groups} \\ 
			&       && \textit{largest}    & \textit{A}   & \textit{B}  & \textit{C}    & \textit{D} \\ \midrule
			&&\multicolumn{1}{c|}{\textit{total}}&\textit{124}&\textit{8}&\textit{5}&\textit{1}&\textit{1}\\ \cline{3-8}
			\multirow{3}{*}{\rotatebox[origin=c]{90}{Ibata et al.}\hspace{.3cm}\rotatebox[origin=c]{90}{groups}} & \textit{globular cluster}   &\multicolumn{1}{c|}{\textit{118}}& 114      & 4     & 0   &0  &0   \\
			& \textit{likely globular cluster} &\multicolumn{1}{c|}{\textit{12}}& 10        &  1 &0 &0     & 1     \\
			& \textit{contaminants} & \multicolumn{1}{c|}{\textit{9}}& 0        &  3 &5  & 1 &0     \\ \bottomrule
		\end{tabular}
	\end{table}
	\def\Ibata{Ibata et al.\xspace} 
	For the purpose of comparison, we report in Table~\ref{tab:VSibata} the confusion matrix of the groups obtained via the DPM-G model against the groups detected by \Ibata 
	All of the 124 stars belonging to the largest group of the partition estimated based on the DPM-G model belong to the groups identified as \textit{globular cluster} or \textit{likely globular cluster} by \Ibata 
	At the same time, out of the nine stars classified as contaminants by \Ibata, the approach based on the DPM-G model assigns none to the globular cluster, three to group A, five stars to group B, which is composed only by stars considered contaminants in \Ibata, and the star of group C,  which shows an extremely small value for the velocity variable. Finally, group D contains only one star, which is not considered a contaminant by \Ibata
	
	\begin{figure}[h]
	\centering
	\begin{subfigure}{.49\textwidth}
	\centering
    	\includegraphics[width=1\textwidth]{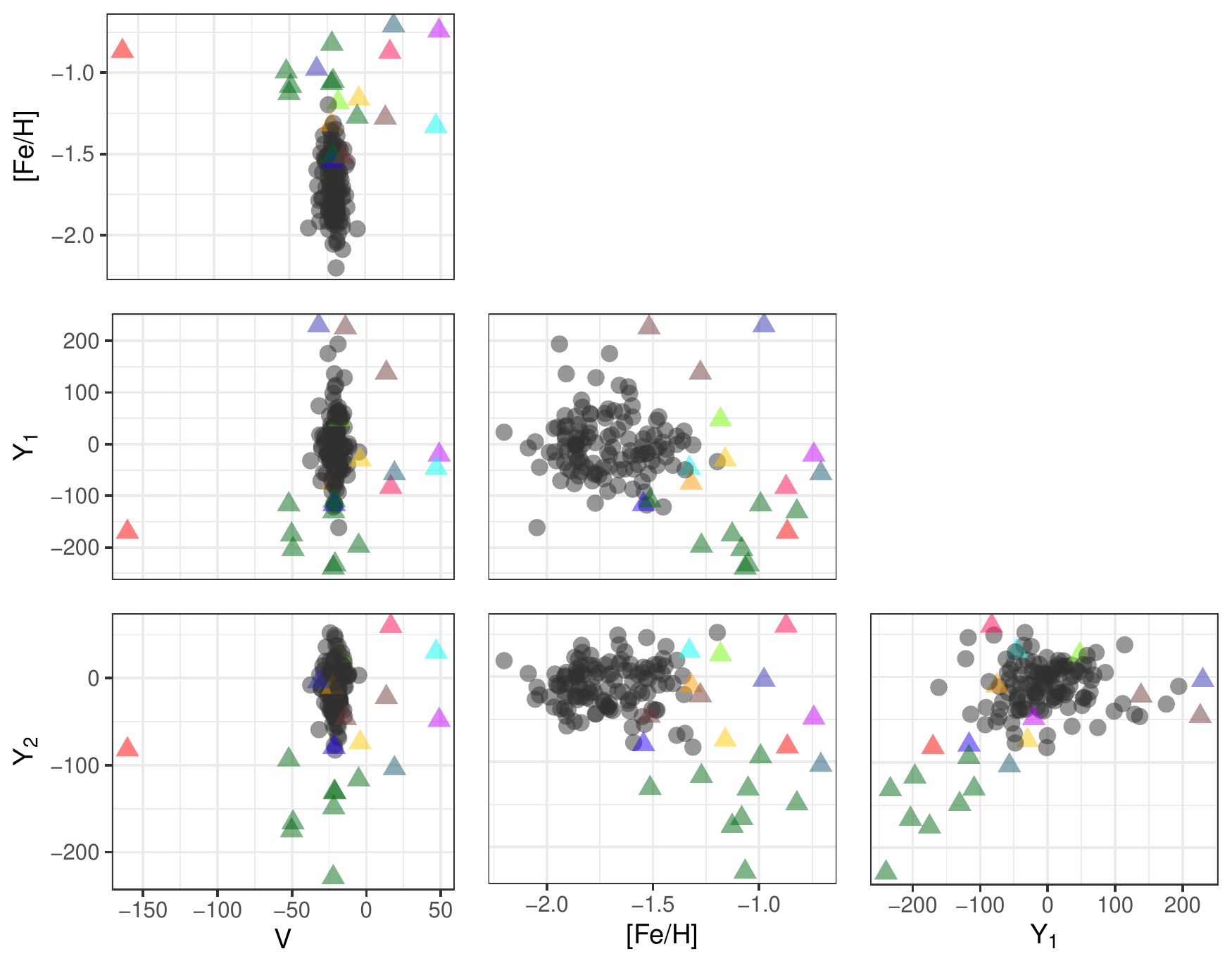} 
    \end{subfigure}
    \begin{subfigure}{.49\textwidth}
	\centering
    	\includegraphics[width=1\textwidth]{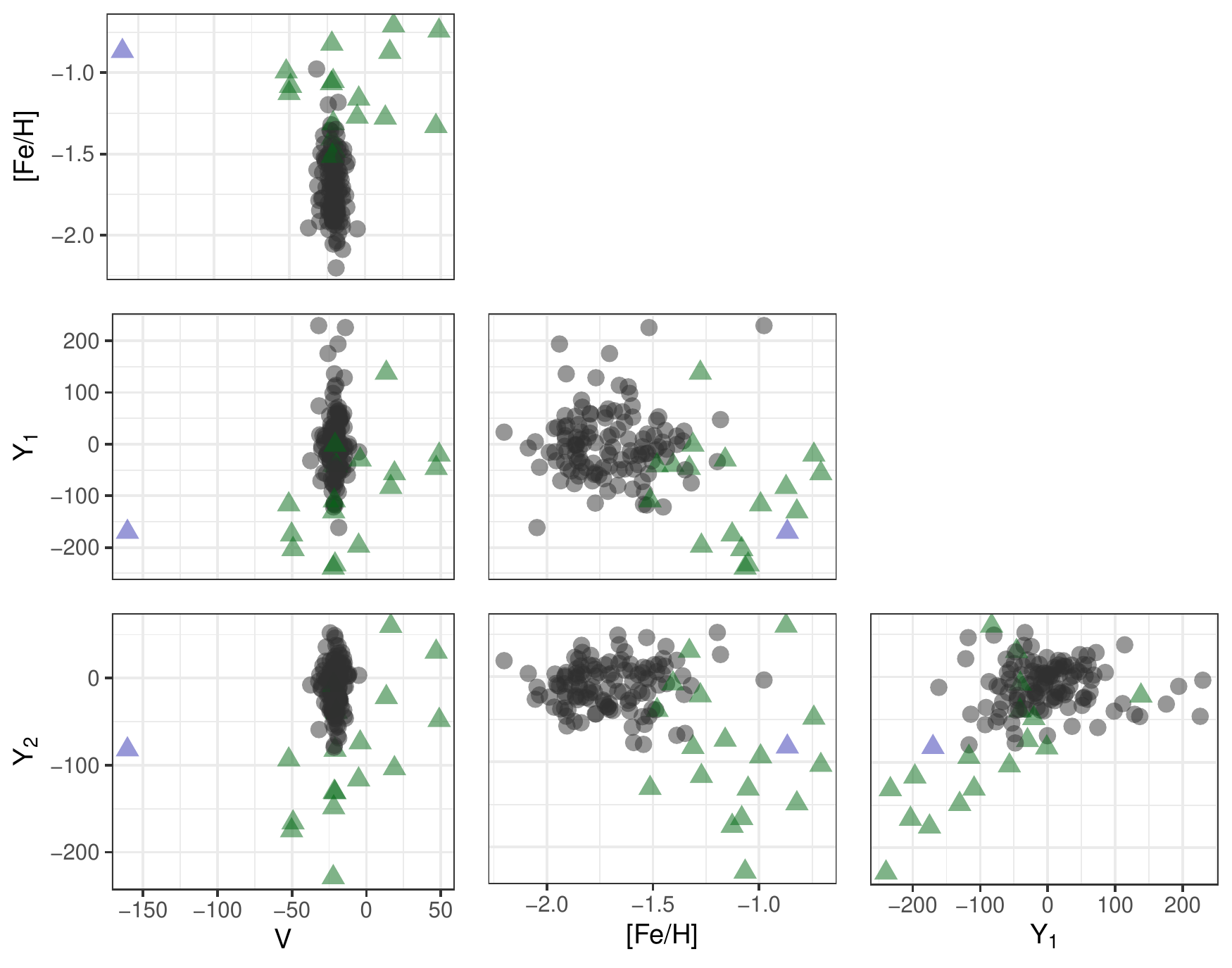} 
    \end{subfigure}\caption{\label{fig:uncert}NGC~2419 data. Lower bound (left) and upper bound (right) of the credible ball on the partitions' space, estimated via DPM-G model combined with \citet{Wad17}'s variation of information method.}%
    \end{figure}
	
	
In order to characterize the uncertainty associated to the estimated optimal partition displayed in Figure~\ref{fig:dens}, we considered a 95\% posterior credible ball in the space of partitions, based on the variation of information metric 
	\citep[see][for details]{Wad17}. Figure~\ref{fig:uncert} shows the vertical lower bound and the vertical upper bound of such credible ball: the first one is the partition in the credible ball which, among those with the largest number of clusters, is the most distant from the optimal partition; the latter one is the partition in the credible ball which, among those with the smallest number of clusters, is the most distant from the optimal one.
	The lower bound displays a total of $13$ groups, with the largest one, object of interest in our analysis, composed by $119$ observations. The upper bound instead is composed of only $3$ groups, with the largest one counting $121$ observations. The largest groups in the two vertical bounds share $116$ observations, thus showing a limited variability, as far as the size of the largest cluster, main object of our analysis, is concerned. This nicely suggests that the adopted procedure for differentiating stars belonging to the globular cluster and contaminants can be considered robust. Finally, further insight on the clustering structure of the data is provided by Figure~\ref{fig:heat}, which shows the heatmap representation of the posterior similarity matrix obtained from the MCMC output. In agreement with the partition obtained by applying the approach of \citet{Wad17},  one main group,  identified with the globular cluster, can be clearly detected in Figure~\ref{fig:heat}. As for the remaining stars, arguably the contaminants, there seem to be two well defined groups, A and B, and a few stars whose group membership is less certain.
	
	\begin{figure}[h]
		\centering
		\includegraphics[width=0.8\textwidth]{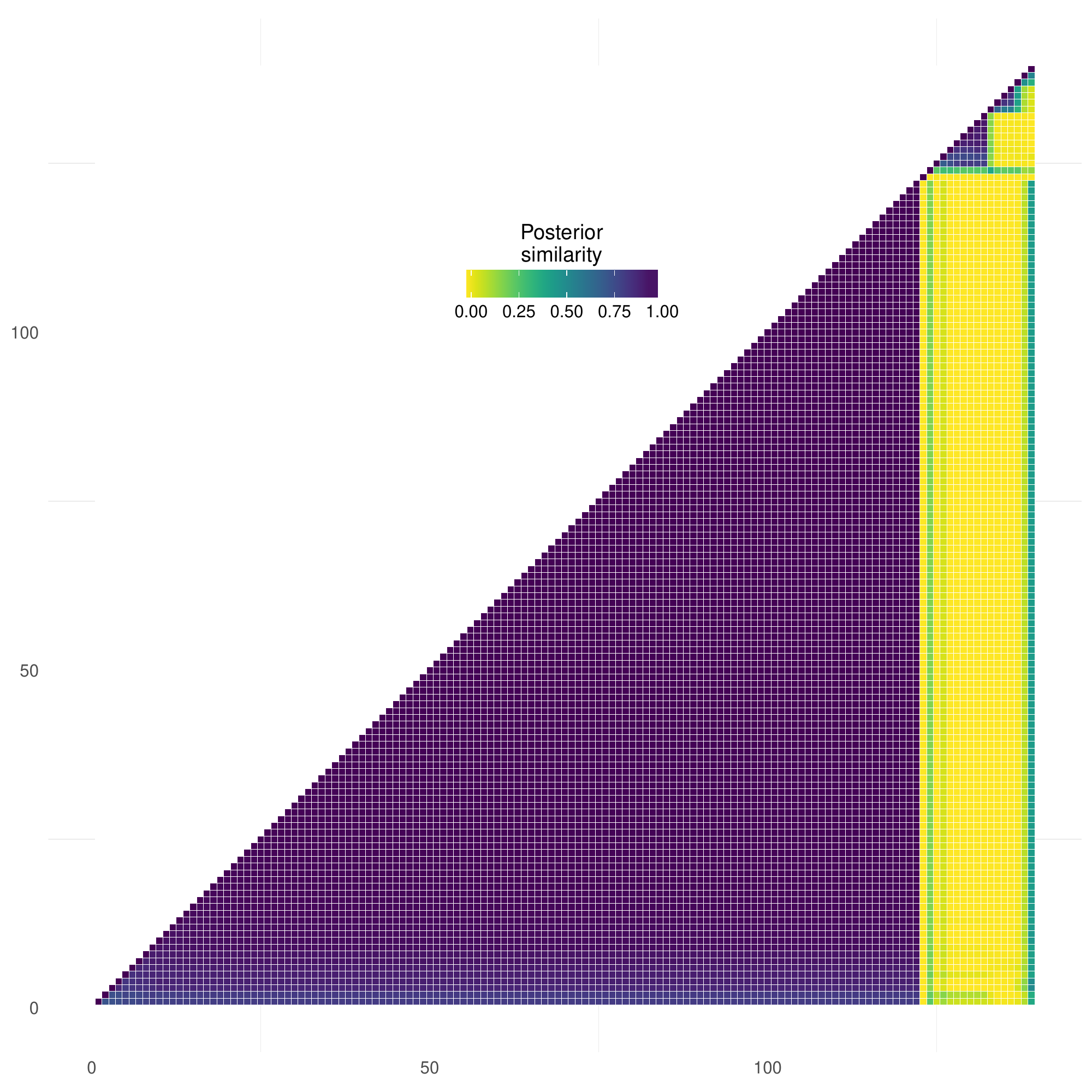} 
		\caption{NGC~2419 data. Heatmap representation of the posterior similarity matrix obtained based on DPM-G model.}\label{fig:heat}
	\end{figure}
	
	\section{Conclusions}\label{sec:conclusions}
	
	The purpose of this paper was to investigate the behaviour of the multivariate DPM-G model when affine transformations are applied to the data. To this end we focused on the DPM-G model with independent normal and inverse-Wishart specification for the base measure.
	Our investigation covered both the finite sample size and the asymptotic framework. Specifically, in Proposition~\ref{prop:exa}, given any affine transformation $g$, an explicit model specification, depending on $g$, was derived so to ensure coherence between posterior inferences carried out based on a dataset or its transformation via $g$. 
	We then considered a different setting where the specification of the model is assumed independent of the specific transformation $g$. In this case, we formalised the notion of asymptotic robustness of a model under transformations of the data and identified mild conditions on the true data generating distributions which are sufficient to ensure that the DPM-G model features such a property. Specifically, Theorem~\ref{thm:asym} shows that the posterior distributions obtained conditionally on a dataset or any affine transformation of it, become more and more similar as the sample size grows to infinity. Inference on densities and, as suggested by the simulation study, on the clustering structure underlying the data, thus becomes increasingly less dependent on the affine transformation applied to the data, as the sample size grows. 
	As a special case, Theorem~\ref{thm:asym} implies that posterior inference based DPM-G models is asymptotically robust to data transformations commonly adopted for the sake of numerical efficiency, such as standardisation or normalisation. This observation is particularly relevant when dealing with the astronomical unsupervised clustering problem motivating this work. 
	Due to the lack of prior information on the dimensional constants relating different physical units, we resorted to a standardisation of each component of the data and chose an arbitrary model specification. Prior information was available in the form of the experts' prior opinion on the expected number of groups in the dataset and was used to elicit the hyperprior distribution for $\alpha$, the precision parameter of the DP. 

\section*{Acknowledgements}

This work was developed in the framework of the Ulysses Program for French-Irish collaborations (43135ZK) and the Grenoble Alpes Data Institute. The authors wish to thank \href{https://www.unibo.it/sitoweb/carlo.nipoti/en}{Carlo Nipoti} for suggesting the motivating astronomical problem, and \href{http://mistis.inrialpes.fr/people/girard/}{St\'ephane Girard} for helpful discussions on the set of assumptions of Theorem~\ref{thm:asym}. 
The authors are also grateful to \href{http://bf-astro.com/index.htm}{Bob Franke} for the picture in Figure~\ref{fig:glob}.

\appendix

\section{Appendix}\label{sec:proofs}
	
	
	\subsection{Proof of Proposition~\ref{prop:exa}}\label{subsec:pro1}
	
		Model $\tilde f_{\piv}$ can be written as
		\begin{align*}
		\tilde f_{\piv}(\xv)&=\int (2\pi)^{-\frac{d}{2}} \det(\Sigmam)^{-\frac{1}{2}} \exp\left\{-\frac{1}{2} (\xv-\muv)^\intercal \Sigmam^{-1} (\xv - \muv)\right\} \tilde{P}(\ddr \muv, \ddr \Sigmam; \piv)\\
		&=\int (2\pi)^{-\frac{d}{2}} |\det(\Cm)| \det(\Cm \Sigmam \Cm^\intercal)^{-\frac{1}{2}}\\
		&\times \exp\left\{-\frac{1}{2} (\Cm \xv+\bv- \Cm \muv-\bv)^\intercal (\Cm \Sigmam \Cm^\intercal)^{-1} (\Cm \xv+\bv - \Cm \muv-\bv)\right\} \tilde{P}(\ddr \muv, \ddr \Sigmam; \piv).
		\end{align*}
		By performing the change of variables $\Sm= \Cm\Sigmam \Cm^\intercal$ and $\mv=\Cm\muv+\bv$ and observing that, by standard properties of the inverse-Wishart and normal distributions, 
		\begin{itemize}
			\item[1.] $\Sigmam \sim IW(\nu_0,\Sm_0)$ implies $\Sm \sim IW (\nu_0,\Cm \Sm_0 \Cm^\intercal)$,
			\item[2.] $\muv \sim N_d(\mv_0,\Bm_0)$ implies $\mv \sim N_d(\Cm \mv_0+\bv,\Cm \Bm_0 \Cm^\intercal)$,
			\item[3.] $\Xv \sim N_d(\muv,\Sigmam)$ implies $\Cm \Xv+\bv \sim N_d(\mv,\Sm)$,\end{itemize}
		we obtain
		\begin{align*}
		\tilde f_{\piv}(\xv)&=|\det(\Cm)|\int (2\pi)^{-\frac{d}{2}} \det(\Sm)^{-\frac{1}{2}}\\
		&\times \exp\left\{-\frac{1}{2} (\Cm \xv+\bv- \mv)^\intercal \Sm^{-1} (\Cm \xv+\bv - \mv)\right\} \tilde{P}(\ddr \mv, \ddr \Sm; \piv_g)\\
		&=|\det(\Cm)| \tilde f_{\piv_g}(g(\xv)).
		\end{align*}
		
		A simple reparametrisation leads to $\tilde f_{\piv_g}=|\det(\Cm)|^{-1}\tilde f_{\piv}\circ  g^{-1}$. All the identities in this proof are deterministic, that is they hold for every $\omega\in\Omega$.


	\subsection{Proof of Theorem~\ref{thm:asym}}
	
	The proof of Theorem~\ref{thm:asym} relies on results proved by \citet{Can17}. 
	We start by deriving a set of simpler conditions implying those of \citet{Can17}.

	\begin{lemma}\label{lem:improve}
		Let $f^*$ be a density function on $\R^d$ that satisfy the conditions of Theorem~\ref{thm:asym}
		\begin{enumerate}
			\item[{\normalfont A1}.] $0 < f^*(\xv) < M$, for some constant $M$ and for all $\xv \in \R^d$,
			\item[{\normalfont A2}.] for some $\eta > 0$, $\int \|\xv\|^{2(1+\eta)}f^*(\xv)\ddr\xv < \infty$,
			\item[{\normalfont A3}.] $\xv\mapsto f^*(\xv)\log^2(\varphi_\delta (\xv))$ is bounded on $\R^d$, where $\varphi_\delta(\xv) = \inf_{\{\tv\,:\,\|\tv-\xv\|<\delta\}}f^*(\tv)$.
		\end{enumerate}
		Then $f^*$ also satisfies
		\begin{enumerate}
			\item[{\normalfont A4}.] $\left|\int f^*(\xv) \log f^*(\xv) \ddr \xv \right| < \infty$,
			\item[{\normalfont A5}.] 	$\exists\,\delta > 0$ such that $\int f^*(\xv) \log\left(f^*(\xv)/\varphi_\delta(\xv)\right)\ddr \xv < \infty.$
		\end{enumerate}
	\end{lemma}
	
	\begin{proof}[Lemma~\ref{lem:improve}]
		We prove that A4 is satisfied by first assuming that $f^*$ is univariate. Since the function $u\mapsto u|\log u|$ is continuous from $\R_+$ to $\R_+$, the function $\xv\mapsto f^*(\xv) |\log f^*(\xv)|$ is bounded by assumption A1. Thus,
		for any $a>0$,
		$$\int_{-a}^a f^*(\xv) \left|\log f^*(\xv)\right| \ddr \xv < \infty.$$
		Then, integrating over the remaining part of the support $R_a=(-\infty,-a)\cup (a,+\infty)$ yields
		\begin{align*}
		\left|\int_{R_a} f^*(\xv) \log f^*(\xv) \ddr \xv \right|
		&\leq
		\int_{R_a} f^*(\xv) \left|\log f^*(\xv) \right| \ddr \xv
		=
		\int_{R_a} \xv \sqrt{f^*(\xv)} \frac{\sqrt{f^*(\xv)}}{\xv} \left|\log f^*(\xv) \right| \ddr \xv
		\\
		&\leq
		\left(\int_{R_a} \|\xv\|^2 f^*(\xv) \ddr\xv \int_{R_a} \|\xv\|^{-2} f^*(\xv) (\log f^*(\xv))^2 \ddr \xv\right)^{1/2},
		\end{align*}
		by Jensen's inequality and Cauchy--Schwarz inequality respectively. The first integral in the right-hand side above is finite by assumption A2. 
		For the same reason as above, the function $\xv\mapsto f^*(\xv) (\log f^*(\xv))^2$ is bounded. Since $\xv \mapsto \|\xv\|^{-2}$ is integrable on $R_a$, the second  integral in the right-hand side above is also finite. 
		In dimension $d$, the same argument holds by applying Cauchy--Schwarz' inequality several times so as to obtain an integrable power  $\|\xv\|^{-p}$, $p>d$, in the second integral.\\ 
		In order to prove A5, we note that 
		\begin{align*}
		\left|\int f^*(\xv) \log\left(f^*(\xv)/\varphi_\delta(\xv)\right)\ddr \xv\right|
		&=
		\left|\int f^*(\xv) \log\left(f^*(\xv)\right)\ddr \xv-
		\int f^*(\xv) \log\left(\varphi_\delta(\xv)\right)\ddr \xv\right|\\
		&\leq
		\left|\int f^*(\xv) \log\left(f^*(\xv)\right)\ddr \xv\right|+
		\int f^*(\xv) \left|\log\left(\varphi_\delta(\xv)\right)\right|\ddr \xv.
		\end{align*}
		We proved that the first  integral in the right-hand side above is finite. The fact that the second integral is also bounded is proved exactly in the same way by using assumption A3.
	\end{proof}

	Let $\lambdav(\Sigmam^{-1}):=(\lambda_1(\Sigmam^{-1}), \dots, \lambda_d(\Sigmam^{-1}))$ be the vector of eigenvalues, in increasing order, of $\Sigmam^{-1}$, the precision matrix of the Gaussian kernel. Henceforth we write $f(x)\lesssim g(x)$ to indicate that the inequality $f(x)\leq b g(x)$ holds for some positive constant $b$ and for any $x$.
	
	\begin{theo}\citep[Theorem 2 in][]{Can17}.\label{th:cons}
		Let $f^* \in \mathscr{F}$, true generating density of $\Xv^{(n)}$, satisfy the conditions stated as assumptions A1, A2, A4 and A5 in Lemma~\ref{lem:improve}. 
		Let model $\Xv^{(n)}$ by means of a DPM-G model defined in \eqref{eq:DPmmls}. 
		Suppose that the base measure $P_0$ has the product form $P_0(\ddr \muv,\ddr \Sigmam)=P_{0,1}(\ddr \muv) P_{0,2}(\ddr \Sigmam)$ and that $P_{0,1}$ and $P_{0,2}$ satisfy the following conditions: for some positive constants $c_1,\, c_2,\, c_3,\, r> (d-1)/2$ and $\kappa > d(d - 1)$,
		\begin{enumerate}
			\item[B1.] $P_{0,1}(\|\muv\|>x) \lesssim x^{-2(r+1)}$,
			\item[B2.] $P_{0,2}(\lambda_d(\Sigmam^{-1})>x) \lesssim \exp\left\{-c_1 x^{c_2}\right\}$,
			\item[B3.] $P_{0,2}\left(\lambda_1(\Sigmam^{-1})< \frac{1}{x}\right) \lesssim x^{-c_3}$,
			\item[B4.] $P_{0,2}\left(\frac{\lambda_d(\Sigmam^{-1})}{\lambda_1(\Sigmam^{-1})} > x\right) \lesssim x^{-\kappa}$,
		\end{enumerate}
		all for any sufficiently large  $x$.  
		Then the posterior distribution $\Pi(\cdot|\Xv^{(n)})$ is consistent at $f^*$, that is, for every $\varepsilon>0$,
		\begin{equation*}
		\Pi \left(f : \rho(f, f^*) < \varepsilon \mid \Xv^{(n)}\right) \longrightarrow 1,
		\end{equation*}
		in $F^*_n$-probability, as $n\to\infty$.
	\end{theo}
	
	Theorem~\ref{th:cons} provides general conditions on the base measure $P_0$ which guarantee consistency of the posterior distribution. The next lemma shows that these conditions are met by the normal/inverse-Wishart base measure \eqref{eq:base}.
	\begin{lemma}\label{lemma1}
		Conditions  B1--B4 of Theorem~\ref{th:cons} are satisfied by the multivariate normal/inverse-Wishart base measure \eqref{eq:base} with $\nu_0>(d + 1)(2d - 3)$.
	\end{lemma}
	Although the proof of Lemma~\ref{lemma1} can be found in \citet{Can17} (their Corollary 1,  relying, in turn, on results by \citet{She13}), we provide it in Appendix~\ref{sec:lem1} for the sake of completeness and in order to account for the slightly different prior specification considered in this work.
	Next  lemma shows that if $f^*$ satisfies conditions A1--A3 of Theorem~\ref{thm:asym}, so does $f^*_g:=|\det(\Cm)|^{-1}f^*\circ g^{-1}$, for any invertible affine transformation $g$.
	\begin{lemma}\label{lem:f0c}
		If conditions A1--A3 of Theorem~\ref{thm:asym} are satisfied by $f^*$, then for any  invertible affine transformation $g(\xv)=\Cm \xv + \bv$, they are also satisfied by $f^*_g=|\det(\Cm)|^{-1}f^*\circ g^{-1}$.
	\end{lemma}
	The proof of Lemma~\ref{lem:f0c} is postponed to Appendix~\ref{sec:lem1}. An analogous result is proved by \citet{shi2018} (see their Lemma 2), although for a different set of assumptions on the true data generating density. 
	We are now ready to prove Theorem~\ref{thm:asym} by combining Theorem~\ref{th:cons} with Lemma~\ref{lem:improve},  Lemma~\ref{lemma1}  and Lemma~\ref{lem:f0c}.
	\begin{proof}[Theorem~\ref{thm:asym}]
		According to Lemma~\ref{lem:improve}, the set of assumptions A1, A2, A4 and A5 (as appearing in Lemma~\ref{lem:improve}) is implied by assumptions A1, A2 and A3 of Theorem~\ref{thm:asym}. So under assumptions A1, A2 and A3, Theorem~\ref{th:cons} holds. 
		By combining it with Lemma~\ref{lemma1} and Lemma~\ref{lem:f0c}, we have that for any $\epsilon>0$, 
		\begin{align}
		&\Pi\left(f:\rho(f,f^*)<\sfrac{\epsilon}{2}\mid \Xv^{(n)}\right)\longrightarrow 1,\label{eq:lim1}\\
		&\Pi\left(f:\rho(f,f_g^*)<\sfrac{\epsilon}{2}\mid g(\Xv)^{(n)}\right)\longrightarrow 1,\label{eq:lim2},
		\end{align}
		both in $F^*_n$-probability, as $n\to\infty$. 
		We notice that the distance $\rho$ is invariant with respect to change of variables and thus  
		$\rho(|\det(\Cm)| f_2\circ g,f^*) = \rho(f_2,f^*_g)$. This, combined with the triangular inequality, leads to 
		\begin{align*}
		&\Pi_2((f_1,f_2): \rho(f_1,|\det(\Cm)|f_2\circ g)<\epsilon\mid \Xv^{(n)})\\
		&\geq \Pi_2\left((f_1,f_2): \rho(f_1,f^*)<\sfrac{\epsilon}{2}, \,\rho(f_2,f_g^*)<\sfrac{\epsilon}{2}\mid \Xv^{(n)}\right)\\
		&\geq \Pi_2\left((f_1,f_2): \rho(f_1,f^*)<\sfrac{\epsilon}{2}\mid \Xv^{(n)}\right) + 
		\Pi_2\left((f_1,f_2): \rho(f_2,f_g^*)<\sfrac{\epsilon}{2}\mid \Xv^{(n)}\right) - 1\\
		&=   \Pi\left(f_1: \rho(f_1,f^*)<\sfrac{\epsilon}{2}\mid \Xv^{(n)}\right) + 
		\Pi\left(f_2: \rho(f_2,f_g^*)<\sfrac{\epsilon}{2}\mid g(\Xv)^{(n)}\right) - 1\\
		&\quad\quad\longrightarrow 1+1-1 =1,
		\end{align*}
		in $F_n^*$-probability, as $n\to\infty$. As a result, 
		\begin{equation*}
		\Pi_2\left((f_1,f_2): \rho(f_1,|\det(\Cm)| f_2 \circ g)<\varepsilon \mid \Xv^{(n)}\right)\longrightarrow 1,
		\end{equation*}
		in $F_{n}^*$-probability, as $n\to\infty$.
	\end{proof}
	
	
	\subsection{Proof of additional lemmas}
	
	\begin{proof}[Lemma~\ref{lemma1}\label{sec:lem1}]
		We check, point-by-point, that the conditions of Theorem~\ref{th:cons} are satisfied.
		\begin{enumerate}
			\item[B1.] Since $\muv \sim N_d(\mv_0, \Bm_0)$, then $\|\muv\|^2 \sim \chi^2_d(\delta)$ where $d$ is the dimension of $\muv$ and $\delta = \|\mv_0\|$ is the non-centrality parameter of the chi-squared distribution. Then, for sufficiently large $x$,
			\begin{equation*}
			P_{0,1}\left(\|\muv\|^2 > x\right) \leq \left(\frac{x}{d} \right)^{\frac{d}{2}} \exp\left\{\frac{d-x}{2}\right\} \lesssim x^{-2(r+1)},
			\end{equation*}
			which holds for $r > (d-1)/2$.
			
			\item[B2.] We know that $\Sigmam \sim IW(\nu_0, \Sm_0)$ and we start by considering the case corresponding to $\Sm_0 = \mathbf{I}_d$, where $\mathbf{I}_d$ denotes the $d$-dimensional identity matrix. It is known that $\Tr(\Sigmam^{-1})\sim \chi^2_{\nu_0 d}$. 
			Thus, for sufficiently large $x$,
			\begin{align*}
			P_{0,2}\left(\lambda_d(\Sigmam^{-1}) > x\right) &\leq P_{0,2}\left(\Tr(\Sigmam^{-1}) > x\right)\\
			&\leq \left(\frac{x}{\nu_0 d} \right)^{\frac{\nu_0 d}{2}} \exp \left\{ \frac{\nu_0 d - x}{2} \right\}\\
			&\lesssim  \exp\left\{-c_1 x^{c_2}\right\},
			\end{align*}
			for some positive constants $c_1$ and $c_2$. This result can be easily generalised to the case $\Sm_0 \neq \mathbf{I}_d$ since $IW(\ddr \Sigmam; \nu_0, \Sm_0) =\Sm_0^{-1}IW(\ddr \Sigmam; \nu_0, \mathbf{I}_d)$.
			\item[B3.] We know that $\Sigmam \sim IW(\nu_0, \Sm_0)$ and we start by supposing that $\Sm_0 = \mathbf{I}_d$. 
			The joint distribution of the eigenvalues $\lambdav\left(\Sigmam^{-1}\right)$ is known to be equal to
			\begin{equation*}
			f_{\lambdav}(x_1, \dots, x_d) = c_{d,\nu_0} \exp\left\{ - \sum_{j=1}^d \frac{x_j}{2}\right\} \prod_{j=1}^d x_j^{\frac{(\nu_0 -d + 1)}{2}} \prod_{j < k}(x_k - x_j),
			\end{equation*}
			for some normalising constant $c_{d, {\nu_0}}$, if $(x_1,\dots,x_d) \in (0, \infty)^d$ is such that $x_1 \leq \dots \leq x_d$, and equal to $0$ otherwise. It is easy to verify that, on the support of $f_{\lambdav}$,
			\begin{equation*}
			\prod_{j<k} (x_k - x_j) \leq \prod_{j<k} x_k = \prod_{k=2}^d x_k^{k-1}.
			\end{equation*}
			The density function of $\lambda_1(\Sigmam^{-1})$ then becomes 
			\begin{align*}
			f_{\lambda_1}(x_1) &= \int \dots \int f_{\lambdav}(x_1, \dots, x_d) \ddr x_2\cdots \ddr x_d\\
			&\leq c_{d, \nu_0} x_1^{\frac{\nu_0-d+1}{2}}\e^{-\frac{x_1}{2}} \prod_{k=2}^d \int \limits_0^\infty x_k^{\frac{\nu_0-d+1}{2} + k - 1} \e^{-\frac{x_k}{2}}\ddr x_k\\
			&= c^{\prime}_{d,\nu_0} x_1^{\frac{\nu_0-d+1}{2}}\exp\left\{-\frac{x_1}{2}\right\},
			\end{align*}
			for some new normalising constant $c^{\prime}_{d,\nu_0}$.
			Then for any $x>0$ we have 
			\begin{equation*}
			P_{0,2}\left(\lambda_1(\Sigmam^{-1})< \frac{1}{x}\right)\leq c'_{d,\nu_0}\int_0^\frac{1}{x} x_1^{\frac{\nu_0 - d + 1}{2}}dx_1 \lesssim x^{-c_3 x} 
			\end{equation*}
			for some constant $c_3$ and sufficiently large $x$. Again, this result can be generalised to the case $\Sm_0 \neq \mathbf{I}_d$ since $IW(\ddr \Sigmam; \nu_0, \Sm_0) =\Sm_0^{-1}IW(\ddr \Sigmam; \nu_0, \mathbf{I}_d)$.
			\item[B4.] We know that $\Sigmam \sim IW(\nu_0, \Sm_0)$ and we start by considering the case corresponding to $\Sm_0 = \mathbf{I}_d$. 
			We define $Z(\Sigmam^{-1}) = \lambda_d(\Sigmam^{-1})/\lambda_1(\Sigmam^{-1})$ and the function $q(\lambdav(\Sigmam^{-1})) = (\lambda_1(\Sigmam^{-1}), \dots, \lambda_{d-1}(\Sigmam^{-1}), Z(\Sigmam^{-1}))$. Let $J_{q^{-1}}$ denote the Jacobian of the inverse of the function $q$, and observe that
			\begin{equation*}
			f_{\lambda_1,\dots, \lambda_{d-1}, Z}(x_1, \dots, x_{d-1}, z) = |J_{q^{-1}}|
			f_{\lambdav}(x_1, \dots, x_{d-1}, x_1 z).	
			\end{equation*}
			Then, by marginalising with respect to the first $d-1$ components, we obtain
			\begin{align*}
			f_{Z}(z)	&=	\int \cdots \int |J_{q^{-1}}|
			f_{\lambdav}(x_1, \dots, x_{d-1}, x_1 z)	\ddr x_1\cdots \ddr x_{d-1} \\
			&=  \int \cdots \int c_{d,\nu_0} \exp\left\{-\sum_{j=1}^{d-1}\frac{x_j}{2} - \frac{x_1 z}{2} \right\} \prod_{j=1}^{d-1}x_j^{\frac{\nu_0 + 1 - d}{2}} (x_1 z)^{\frac{\nu_0 + 1 - d}{2}}\\
			& \times  \prod_{j<k\leq d-1}(x_k-x_j) \prod_{j=1}^{d-1}(x_1 z - x_j) x_1 \ddr x_1 \cdots \ddr x_{d-1} \\
			& \leq \int \cdots \int  c_{d,\nu_0} \exp\left\{-\sum_{j=1}^{d-1}\frac{x_j}{2} - \frac{x_1 z}{2} \right\} \prod_{j=1}^{d-1}x_j^{\frac{\nu_0 + 1 - d}{2}} (x_1 z)^{\frac{\nu_0 + 1 - d}{2}}\\
			&\times \prod_{k=2}^{d-1}x_k^{k-1} \prod_{j=1}^{d-1}(x_1 z) x_1 \ddr x_1 \cdots \ddr x_{d-1} \\
			&=c_{d,\nu_0}^\prime z^{(\nu_0 + d - 1)/2}\int \exp\left\{ -x_1 \left(\frac{z+1}{2}\right) \right\} x_1^{\nu_0 + 1}\ddr x_1 \\
			&=c_{d,\nu_0}^\prime (\nu_0 + 1)! \left(\frac{2}{z + 1}\right)^{\nu_0 + 2} z^{(\nu_0 + d - 1)/2} \\
			&= c_{d, \nu_0}^{\prime\prime} \frac{z^{(\nu_0 + d - 1)/2}}{(z+1)^{\nu_0 + 2}} \\
			&\leq  c_{d,\nu_0}'' z^{-(\nu_0-d+5)/2},
			\end{align*}
			for some constants $c_{d,\nu_0}$, $c_{d,\nu_0}'$ and $c_{d,\nu_0}''$. 
			Thus we have 
			\begin{align*}
			P_{0,2}\left(Z>x\right)=&\int_x^\infty f_Z(z)\ddr z
			\leq c_{d,\nu_0}^{\prime\prime} \int_{x}^\infty z^{-(\nu_0-d+5)/2}\ddr z
			\lesssim x^{-\kappa},
			\end{align*}
			for sufficiently large $x$, where $\kappa = (\nu_0-d+3)/2>d(d+1)$ by  the assumption that  $\nu_0 > (d+1)(2d-3)$.
		\end{enumerate}
	\end{proof}

	\begin{proof}[Lemma~\ref{lem:f0c}\label{sec:lem2}]
		We assume that $f^*$ satisfies conditions A1--A3 of Theorem~\ref{thm:asym} and check that the same holds for $f^*_g$.
		\begin{enumerate}[topsep=2ex,label=(\roman*)]
			\item[A1.] Assume that $0 < f^*(\xv) < M$ for every $\xv \in\R^d$ and some $M>0$. Then, for every $\xv\in\R^d$, we have $f^*_g(\xv)= |\det(\Cm)|^{-1}f^* (g^{-1}(\xv))$ which implies 
			\begin{equation*}
			0<f^*_g(\xv) < M^\prime= |\det(\Cm)|^{-1}M.
			\end{equation*}
			\item[A2.] Observe that
			\begin{align*}
			\int \|\xv\|^{2(1+\eta)}f_g^*(\xv) \ddr \xv&=\int \|g(\yv)\|^{2(1+\eta)}f_g^*(g(\yv))|\det(\Cm)| \ddr \yv\\
			&=\int \|g(\yv)\|^{2(1+\eta)}f^*(\yv) \ddr \yv\\
			&\leq \int 2^{2(1+\eta)-1} \left(\|\Cm \yv\|^{2(1+\eta)}+\|\bv\|^{2(1+\eta)}\right)f^*(\yv)\ddr \yv,
			\end{align*}
			where the last inequality follows by combining triangular and Jensen's inequalities. Thus we can write
			\begin{align*}
			\int &\|\xv\|^{2(1+\eta)}f_g^*(\xv) \ddr \xv\\
			&\leq 2^{2(1+\eta)-1}\left(|\det(\Cm)|^{2(1+\eta)}\int \| \yv\|^{2(1+\eta)}f^*(\yv)\ddr \yv+\|\bv\|^{2(1+\eta)}\right)<\infty,
			\end{align*}
			where the last inequality follows by assumption A2 on $f^*$.
			\item[A3.] Since function $g$ is a linear invertible transform, the boundedness of the function $\xv\mapsto f^*(\xv)\log^2(\varphi_\delta (\xv))$ carries over to its counterpart defined with the transformed density $f_g^*$.
		\end{enumerate}
	\end{proof}

	\bibliographystyle{apalike}

\end{document}